\documentclass[aps,prl,reprint,groupedaddress]{revtex4-2}
\usepackage{epsfig,amssymb,amsmath,mathrsfs,amsthm,graphicx,float,color,hyperref}

\usepackage{dcolumn}
\usepackage{bm}
\usepackage{multirow}
\usepackage{enumitem}
\usepackage{MnSymbol}
\usepackage{physics}
\usepackage[ruled,vlined]{algorithm2e}

\def\<{\langle}\def\>{\rangle}
\def\bbra#1{\llangle#1\rvert}
\def\kett#1{\lvert#1\rrangle}
\def\bbrakett#1#2{\llangle#1|#2\rrangle}

\DeclareMathOperator*{\argmin}{arg\,min}
\theoremstyle{plain}
\newtheorem{theorem}{Theorem}
\newtheorem{lemma}{Lemma}

\theoremstyle{definition}

\theoremstyle{remark}

\begin{document} 
\title{Heisenberg-Limited Quantum Metrology without Ancillae}
\author{Qiushi Liu}
\email{qsliu@cs.hku.hk}
\affiliation{
 QICI Quantum Information and Computation Initiative, Department of Computer Science, School of Computing and Data Science, The University of Hong Kong, Pokfulam Road, Hong Kong, China
}
\author{Yuxiang Yang}
\email{yuxiang@cs.hku.hk}
\affiliation{
 QICI Quantum Information and Computation Initiative, Department of Computer Science, School of Computing and Data Science, The University of Hong Kong, Pokfulam Road, Hong Kong, China
}
\date{\today}

\begin{abstract}
Extensive research has been dedicated to the asymptotic theory of quantum metrology, where the goal is to determine the ultimate precision limit of quantum channel estimation when many accesses to the channel are allowed. The ultimate limit has been well established, but in general a noiseless and controllable ancilla is required for attaining it. Little is known about the metrological performance without noiseless ancillae, which is more relevant in practical circumstances. In this Letter, we present a novel theoretical framework to address this problem, bridging quantum metrology and the asymptotic theory of quantum channels. Leveraging this framework, we prove sufficient conditions for achieving the Heisenberg limit with repeated applications of the channel to be estimated, both with and without applying interleaved unitary control operations. For the latter case, we design an algorithm to identify the control operation explicitly. 
\end{abstract}

\maketitle

\emph{Introduction}---Quantum metrology \cite{giovannetti2004quantum,Giovannetti2006PRL,Degen17RMP} is one of the most promising quantum technologies in and beyond the noisy intermediate-scale quantum era \cite{Preskill2018quantumcomputingin}. Its power has been demonstrated in a wide range of contexts, including gravitational wave detection \cite{Schnabel2010,LIGO19PRL}, quantum clocks \cite{Buzek1999PRL,Pedrozo-Penafiel2020}, and high-resolution imaging \cite{Brida2010,LeSage2013}. 

In a prototypical setup, one would like to estimate a parameter $\theta$, given $N$ queries to an unknown quantum channel $\mathcal E_\theta$ in each round of the experiment. By leveraging quantum entanglement in preparing a probe state (parallel strategy) \cite{giovannetti2004quantum,Giovannetti2006PRL} or coherently applying a sequence of $\mathcal E_\theta$ interleaved with control operations (sequential strategy) \cite{Giovannetti2006PRL,vanDam07PRL,Demkowicz14PRL}, the ultimate limit of mean squared error (MSE), $\delta \hat{\theta}^2$, is the Heisenberg limit (HL), $1/N^2$. The $1/N^2$ scaling is known to hold for unitary channel estimation \cite{Yuan15PRL} but is susceptible to noise \cite{Demkowicz-Dobrzanski2012}, which sometimes results in the standard quantum limit (SQL) $1/N$, hindering the quantum scaling advantage. It is thus a desideratum to investigate under what types of noise the HL is still achievable.

Extensive research has focused on identifying the ultimate precision limit for quantum channel estimation \cite{Fujiwara2008,Escher2011,Demkowicz-Dobrzanski2012,Kolodynski_2013NJP,Demkowicz14PRL,Sekatski2017quantummetrology,Demkowicz-Dobrza2017PRX,Zhou2018NC,Zhou2021PRXQ,Kurdzialek23PRL}. The optimal precision follows either the HL or the SQL, depending on whether the channel to be estimated satisfies the ``Hamiltonian-not-in-Kraus-span'' (HNKS) condition \cite{Zhou2021PRXQ}. However, the attainability of optimal limits generally assumes access to unbounded quantum memory when $N \rightarrow \infty$. When the HL can be achieved with sequential strategies, there is no need for the unbounded memory, but noiseless ancillae and syndrome measurements are required for repeatedly applying quantum error correction (QEC) \cite{Demkowicz-Dobrza2017PRX,Zhou2018NC,Zhou2021PRXQ}.

The requirement of noiseless ancillae, however, is a major obstacle in many real-world applications. To circumvent this issue, ancilla-free probe states have been proposed in some metrological strategies \cite{Layden19PRL,Zhou24PRAachieving}, which, nevertheless, either require highly entangled many-body probe states or involve real-time syndrome measurements and delicate QEC operations that introduce additional ancillae. Such requirements are not only challenging for near-term devices \cite{Hou2019PRL,Hou2021PRL}, but also impose strong demands on the programmability of quantum sensors even in the long term. 

In this Letter, we establish a theoretical framework for single-parameter quantum channel estimation with ancilla-free sequential strategies and uncover a close connection between metrological limits and the spectral properties of quantum channels. In contrast to existing results that rely on QEC, we apply our framework to identify sufficient conditions for achieving the \emph{ancilla-free HL} \footnote{In this work we only care about the precision scaling $1/N^2$, regardless of the coefficient.}, with either control-free strategies or identical unitary control operations that can be systematically identified by an algorithm. We summarize our main results in Table~\ref{tab:summary}. The strategies we consider can be implemented by the simple setup in Fig.~\ref{fig:setup}, motivated by typical optical experiments, where it is favorable to apply identical unitary control operations in an optical loop \cite{Hou2019PRL}. 

\begin{table}[!htbp]
\caption{\label{tab:summary}
Comparison of main results with previous works.
}
\begin{ruledtabular}
\begin{tabular}{lll}
\textrm{Scaling}&
\textrm{Condition}&
\textrm{Strategy}\\
\colrule
SQL \cite{Demkowicz-Dobrza2017PRX,Zhou2018NC,Zhou2021PRXQ} & HNKS violated & QEC, unbounded ancilla\\
HL \cite{Demkowicz-Dobrza2017PRX,Zhou2018NC,Zhou2021PRXQ} & HNKS satisfied & QEC, bounded ancilla\\
\multirow{2}{*}{Ancilla-free HL} & Theorem~\ref{thm:condition for HL} & Control-free \\
& Theorem~\ref{thm:HL with unitary control}\footnote{See Theorem~\ref{thm:HL Pauli noise} for a special case in terms of parameter estimation under Pauli noise.} & Unitary control\\
\end{tabular}
\end{ruledtabular}
\end{table}

\begin{figure}[!htbp]
    \centering
    \includegraphics[width=0.3\textwidth]{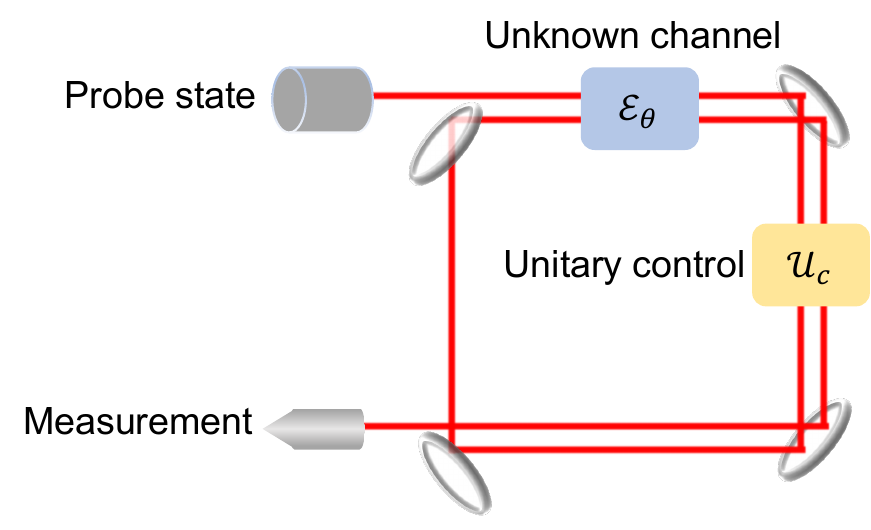}
    \caption{Control-enhanced sequential strategy for estimating $\theta$ from $N$ queries to $\mathcal E_\theta$ ($N=2$ as an illustrated example). Trivial $\mathcal U_c=\mathcal I$ corresponds to a control-free strategy.}
    \label{fig:setup}
\end{figure}

\emph{Lower bound on quantum Fisher information}---Given a single-parameter quantum state $\rho_\theta$, the MSE $\delta \hat{\theta}^2:=\mathrm E\left[(\hat{\theta}-\theta)^2\right]$ for any unbiased estimator $\hat{\theta}$ is restricted by the quantum Cramér-Rao bound \cite{helstrom1976quantum,holevo2011probabilistic,Braunstein1994PRL} $\delta \theta^2 \ge 1/[\nu F^Q(\rho_\theta)]$, where $F^Q(\rho_\theta)$ is the quantum Fisher information (QFI) and $\nu$ is the number of times the experiment is repeated. As the bound is achievable in the limit of large $\nu$, the QFI can be regarded as a score function for single-parameter quantum metrology, quantifying the sensitivity of $\rho_\theta$ to $\theta$. Note that QFI is relevant for \emph{local estimation} in the neighborhood of the true value $\theta=\theta_0$ \footnote{Certainly, the true value $\theta_0$ should be only roughly known before the experiment. In practice, one can adaptively update the metrological strategy based on the previous estimate $\hat{\theta}$ during the experiment \cite{O_E_Barndorff-Nielsen_2000,Hayashi2011CMP,Yang2019CMP}.}. In quantum channel estimation, we focus on the QFI of the output state $\rho_\theta$ of protocols that query the channel $\mathcal E_\theta$ $N$ times. We say the HL (SQL) is achieved if the QFI scales as $N^2$ ($N$), as $N\rightarrow \infty$.

The QFI is defined as $F^Q(\rho_\theta):=\Tr(\rho_\theta L_{\rho_\theta}^2)$,
where $L_{\rho_\theta}$ is the symmetric logarithmic derivative (SLD) defined through $2\dot \rho_\theta = \rho_\theta L_{\rho_\theta} + L_{\rho_\theta} \rho_\theta$, having denoted the derivative of $X$ with respect to $\theta$ by $\dot{X}$. For a pure state $\rho_\theta=\dyad{\psi_\theta}$, the QFI can be computed in a simpler way: 
\begin{equation} \label{eq:QFI pure state}
    F^Q\left({\dyad{\psi_\theta}}\right) = 4\left(\braket{\dot{\psi}_\theta}-\left\lvert\braket{\psi_\theta}{\dot{\psi}_\theta}\right\rvert^2 \right).
\end{equation}

In general, however, it can be challenging to analytically evaluate the QFI of $N$ queries to the channel $\mathcal E_\theta$ for large $N$, as computing the SLD $L_{\rho_\theta}$ requires the spectral decomposition of $\rho_\theta$ \cite{Braunstein1994PRL}. It is thus more practical to obtain a lower bound on the QFI, which already suffices to provide a guaranteed precision scaling. To this end, we employ a technique based on vectorization to derive such an efficiently computable lower bound.

We first bound the state QFI $F^Q(\rho_\theta)$ and then investigate its scaling with $N$. Here, it is useful to represent a quantum state $\rho$ on $d$-dimensional Hilbert space $\mathcal H$ as a $d^2$-dimensional vector $\kett{\rho}$, where $\kett{\rho} = \sum_{ij}\mel{i}{\rho}{j}\ket{i}\ket{j}$ denotes the vectorization of $\rho$ in the computational basis $\{\ket{i}\}_i$. Ref.~\cite{Alipour14PRL} showed that $F^Q\left(\rho_\theta\right)$ has a close relation to the \emph{associated QFI}, 
\begin{equation}
    \tilde F^Q\left(\tilde \rho_\theta\right)=\Tr\left(\tilde \rho_\theta L_{\tilde \rho_\theta}^2\right),
\end{equation}
where $\tilde \rho_\theta:=\kett{\rho_\theta}\bbra{\rho_\theta}/\Tr \left(\rho_\theta^2\right)$ is the density operator of the \emph{associated state}, and $L_{\tilde \rho_\theta}$ is the \emph{associated SLD} defined by $2\dot{\tilde \rho}_\theta=\tilde \rho_\theta L_{\tilde \rho_\theta} + L_{\tilde \rho_\theta} \tilde \rho_\theta$. Specifically, the QFI of $\rho_\theta$ can be lower bounded by the associated QFI up to a multiplicative factor \cite{Alipour14PRL}:
\begin{equation}
    F^Q\left(\rho_\theta\right) \ge \frac{\Tr\left(\rho_\theta^2\right)}{4\lambda_{\mathrm{max}}\left(\rho_\theta\right)} \tilde F^Q \left(\tilde \rho_\theta\right),
\end{equation}
where $\lambda_{\mathrm{max}}\left(\rho_\theta\right)$ is the largest eigenvalue of $\rho_\theta$. Crucially, the prefactor is bounded in an $N$-independent interval: $(\sqrt{d}-1)/[2(d-1)] \le \Tr\left(\rho_\theta^2\right)/[4\lambda_{\mathrm{max}}(\rho_\theta)] \le 1/4$.

In this Letter, we focus on the scaling of the associated QFI $\tilde F^Q \left(\tilde \rho_\theta\right)$, which is easier to compute and serves as a sufficient condition for attaining the same scaling in terms of the original QFI $F^Q(\rho_\theta)$. As a starting point, we slightly generalize the result in Ref.~\cite{Alipour14PRL} and introduce a lemma providing a closed-form formula for $\tilde F^Q \left(\tilde \rho_\theta\right)$ (see the proof in Ref.~\cite{SM_note}).
\begin{lemma} \label{lem:associated QFI}
    The associated QFI of the associated state $\tilde \rho_\theta=\kett{\rho_\theta}\bbra{\rho_\theta}/\Tr \left(\rho_\theta^2\right)$ is given by
    \begin{equation} \label{eq:associated QFI state}
        \tilde F^Q \left(\tilde \rho_\theta\right) = 4 \left\{\frac{\bbrakett{\dot \rho_\theta}{\dot \rho_\theta}}{\Tr\left(\rho_\theta^2\right)} - \left[\frac{\bbrakett{\rho_\theta}{\dot \rho_\theta}}{\Tr\left(\rho_\theta^2\right)}\right]^2\right\}.
    \end{equation}
\end{lemma}
Note that Eq.~(\ref{eq:associated QFI state}) takes a simple form analogous to the pure state QFI given by Eq.~(\ref{eq:QFI pure state}). 

\emph{Asymptotics of quantum channels}---Before delving into the scaling of $\tilde F^Q \left(\tilde \rho_\theta\right)$, let us briefly recall some spectral properties of quantum channels. A quantum channel, i.e., a completely positive trace-preserving map, $\mathcal E_\theta$, can be described by Kraus operators: $\mathcal E_\theta(\rho)=\sum_{i=1}^r K_i^\theta \rho K_i^{\theta\dagger}$. When treated as a superoperator, a quantum channel acting on a $d$-dimensional system has a \emph{Liouville representation} \cite{gilchrist2011vectorizationquantumoperationsuse,Wood15QIC} given by a $d^2 \times d^2$ transition matrix \footnote{Here we consider only quantum channels with equal input and output spaces.}:
\begin{equation} \label{eq:Liouville representation}
    T_\theta:=\sum_{i=1}^r K_i^\theta \otimes K_i^{\theta*}.
\end{equation}
Now we can express the action of the quantum channel on a state as the multiplication of a matrix by a vector, i.e., we have $\kett{\mathcal E_\theta(\rho)}=T_\theta \kett{\rho}$. 

We focus on the spectral properties of $T_\theta$ (see Refs.~\cite{wolf2012quantum,Burgarth_2013,watrous2018theory} for further details). Consider the spectrum $\{\lambda_i\}_i$ of $T_\theta$, such that $T_\theta\kett{R_i}=\lambda_i\kett{R_i}$. For clarity, we call $\{\kett{R_i}\}_i$ the \emph{eigenvectors} and $\{R_i\}_i$ the \emph{eigenmatrices}. All the eigenvalues satisfy $|\lambda_i|\le 1$ (with at least one $\lambda_j=1$ associated with the \emph{fixed point state} $\rho_*$), and those satisfying $|\lambda_i|= 1$ are called \emph{peripheral eigenvalues}, with the corresponding eigenvectors called \emph{peripheral eigenvectors}. 

We denote the set of indices of peripheral eigenvalues by $\mathsf{S}$. In general, $T_\theta$ is not necessarily diagonalizable, but it may admit only a Jordan normal form. However, for the peripheral eigenvalues $\{\lambda_i\mid i\in\mathsf{S}\}$, the Jordan blocks are all one-dimensional \cite[Proposition 6.2]{wolf2012quantum}. In the asymptotic limit of applying $N$ quantum channels repeatedly, i.e., applying $T_\theta^N$, only the peripheral eigenvalues and eigenvectors take effect in the output $\lim_{N\rightarrow\infty}T_\theta^N\kett{\rho}$, which could greatly facilitate our later analysis. 

\emph{Achieving the HL without control}---Now, we first consider applying the quantum channel $\mathcal E_\theta$ $N$ times sequentially without any control. The goal is thus to evaluate the QFI of $\rho_\theta=\mathcal E_\theta^N(\rho_0)$ for a suitable probe state $\rho_0$. With the vectorization trick, we focus on the associated QFI for $T_\theta^N\kett{\rho_0}$ to obtain a lower bound on $F^Q\left[\mathcal E_\theta^N\left(\rho_0\right)\right]$. Leveraging Lemma~\ref{lem:associated QFI}, we can connect the asymptotic associated QFI to the spectral properties of $T_\theta$ as follows.

\begin{lemma}
\label{lem:asymptotic associated QFI}
    Consider a quantum channel $T_\theta$ defined by Eq.~(\ref{eq:Liouville representation}) with peripheral eigenvalues $\{\lambda_i \mid i\in \mathsf{S}\}$ and peripheral eigenvectors $\{\kett{R_i}\mid i\in \mathsf{S}\}$. For any input state $\kett{\rho_0}=\sum_{i=1}^{d^2}a_{i}\kett{R_i}$ (having chosen a set of linearly independent vectors $\{\kett{R_i}\}_{i=1}^{d^2}$ that forms a Jordan basis \footnote{The choices of $\kett{R_i}$ coincide with the peripheral eigenvectors of $T_\theta$ when $i\in \mathsf S$.}), after $N\rightarrow\infty$ repeated application of $T_\theta$, the asymptotic associated QFI of $\tilde \rho_\theta=\kett{\rho_\theta}\bbra{\rho_\theta}/\Tr \left(\rho_\theta^2\right)$ for $\kett{\rho_\theta}=T_\theta^N\kett{\rho_0}$ is given by 
    \begin{equation} \label{eq:asymptotic associated QFI}
        \lim_{N\rightarrow\infty}\frac{\tilde F^Q \left(\tilde \rho_\theta\right)}{N^2} = 4 \left[\sum_{i,j\in\mathsf S}\beta_{ij}\frac{\dot{\lambda}_i^*\dot{\lambda}_j}{\lambda_i^*\lambda_j} - \left(\sum_{i,j\in\mathsf S}\beta_{ij}\frac{\dot{\lambda}_j}{\lambda_j}\right)^2\right],
    \end{equation}
where $\beta_{ij}=(\lambda_i^*\lambda_j)^Na_ia_j\bbrakett{R_i}{R_j}/\Tr\left(\rho_\theta^2\right)$ satisfying $\sum_{i,j\in\mathsf S}\beta_{ij}=1$ and $\beta_{ij}=\beta_{ji}^*$ \footnote{We assume that $\lambda_i, \kett{R_i}, a_i$ are all continuously differentiable with respect to $\theta$.}.
\end{lemma}

Lemma~\ref{lem:asymptotic associated QFI} straightforwardly yields a sufficient condition for achieving the HL without control, formulated as the existence of the nonzero derivative of a peripheral eigenvalue, as stated in Theorem~\ref{thm:condition for HL} (see the proofs of Lemma~\ref{lem:asymptotic associated QFI} and Theorem~\ref{thm:condition for HL} in Ref.~\cite{SM_note}).

\begin{theorem}\label{thm:condition for HL}
    If, at $\theta=\theta_0$, there exists some peripheral eigenvalue $\lambda_j$ such that $\dot \lambda_j\neq 0$, then the output state QFI achieves the HL by the repeated application of $T_\theta$ without control. The input state can be taken as $\rho_0=(1-\alpha)I/d + \alpha \rho_* +\beta\left(R_j + R_j^\dagger\right) $, where $\alpha, \beta >0$ are arbitrary parameters for $\rho_0$ to be a legitimate density matrix and $R_j$ is the peripheral eigenmatrix associated with $\lambda_j$.
\end{theorem}

\emph{Achieving the HL with unitary control}---The simple condition provided by Theorem~\ref{thm:condition for HL} does not hold in some metrological scenarios. For instance, the HL cannot be achieved without control for estimating $\theta$ given $N$ queries to a unitary channel $\mathcal U_\theta(\cdot) = U_\theta \cdot U_\theta^\dagger$, where $U_\theta=e^{-\mathrm i (\cos \theta \sigma_x + \sin \theta \sigma_z)}$ \cite{Pang14PRA,Liu2015SR,Yuan15PRL}. One can check that all the peripheral eigenvalues of $U_\theta \otimes U_\theta^*$ have vanishing derivatives, as only the eigenbasis changes with $\theta$. Nevertheless, we can apply a unitary control $U_{\theta_0}^\dagger$ following each query to $U_\theta$ to recover the HL at $\theta=\theta_0$. 

It is thus natural to ask under what circumstances unitary control can be applied to achieve the HL in noisy metrology, which remains largely unexplored. Building on Theorem~\ref{thm:condition for HL}, we establish a sufficient condition for achieving the HL using interleaved identical unitary control.

\begin{theorem} \label{thm:HL with unitary control}
    Consider estimating a quantum channel $T_\theta$ defined by Eq.~(\ref{eq:Liouville representation}). We denote by $P$ the projection on the subspace $\mathcal P$ of eigenvectors of $T_\theta^\dagger T_\theta$ with eigenvalues $1$.  
    If, at $\theta=\theta_0$, (i)  
    \begin{equation} \label{eq:nonvanishing signal}
        P T_\theta^\dagger\dot T_\theta P \neq 0\quad\text{(nonvanishing signal)}
    \end{equation}
    and (ii) there exists a unitary $U_c$ such that for some eigenmatrix $R_0$ of $P T_\theta^\dagger \dot{T}_\theta P$ associated with a nonzero eigenvalue,
    \begin{equation} \label{eq:unitary equivalence}
        U_c^\dagger R_0 U_c =  \sum_kK_k^{\theta}R_0K_k^{\theta\dagger} \quad\text{(partially reversible)}
    \end{equation}
 then the output state QFI achieves the HL by applying the same unitary control $U_c$ following each channel $T_\theta$. 
\end{theorem}
The proof is provided in Ref.~\cite{SM_note}. The two conditions in Theorem~\ref{thm:HL with unitary control} have intuitive interpretations. The first condition indicates that the signal operator $T_\theta^\dagger \dot{T}_\theta$ has a nontrivial effect on the subspace $\mathcal P$, and the second condition ensures that the effective action of the quantum channel on $R_0$ is reversible, simulated by a unitary operation. By appending the unitary control $U_c$ to $T_\theta$, the strategy reduces to the repeated application of $(U_c \otimes U_c^*) T_\theta$, and therefore, Theorem~\ref{thm:condition for HL} is applicable.

\emph{Sketch of an algorithm for unitary control}---We further design a numerical algorithm (see Algorithm~\ref{alg:unitary control} in Ref.~\cite{SM_note}) that can successfully identify a unitary control operation $U_c$ if and only if Eq.~(\ref{eq:unitary equivalence}) holds. Note that, except for the trivial case where $R_0$ and $R_0^\dagger$ are associated with the same eigenvalue of $PT_\theta^\dagger \dot{T}_\theta P$ (such that $R_0$ can be chosen to be Hermitian and finding $U_c$ is trivial), Eq.~(\ref{eq:unitary equivalence}) can be reformulated as
\begin{equation}
    U_c^\dagger R_i U_c =  \sum_kK_k^{\theta_0}R_iK_k^{\theta_0\dagger},\ \forall i=1,2,
\end{equation}
where $R_1=R_0+R_0^\dagger$ and $R_2=\mathrm i (R_0-R_0^\dagger)$ are both Hermitian. The key idea behind the algorithm is using the principal angles \cite{jordan1875essai,Bjorck1973,Galantai06Jordan} to characterize the relation between subspaces. For two arbitrary subspaces $F$ and $G$ associated with projections $\Pi_F$ and $\Pi_G$ on the subspaces, the principal angles can be determined by performing the singular value decomposition of $\Pi_F \Pi_G$ (i.e., the principal angles are arccosine values of the singular values). The geometrical relationship between two subspaces can then be fully characterized by the relationships between a set of lower-dimensional subspaces. By iteratively reducing the dimensions, the final reduced subspaces either become mutually orthogonal or share equal dimensions with degenerate principal angles. It is thus possible to select canonical orthonormal bases for each subspace such that their inner products are uniquely determined. If these inner products are preserved under the action of $T_\theta$, then $U_c$ satisfying Eq.~(\ref{eq:unitary equivalence}) can be identified, and vice versa. A formal description and proof of the algorithm are provided in Ref.~\cite{SM_note}, with the code implementation openly available on GitHub \cite{github}.

\emph{Achieving the HL under Pauli noise}---For parameter estimation under Pauli noise, a widely considered scenario in quantum metrology \cite{Sekatski2017quantummetrology,Escher2011,Demkowicz-Dobrzanski2012,Kolodynski_2013NJP,Demkowicz14PRL,Arrad14PRL,Kessler2014PRL,ozeri2013heisenberg,Duer2014PRL,Demkowicz-Dobrza2017PRX,Zhou2018NC,Zhou2021PRXQ,Kurdzialek23PRL}, our framework yields a convenient criterion for achieving the ancilla-free HL. 

Consider estimating an $n$-qubit channel $T_\theta$ defined by Eq.~(\ref{eq:Liouville representation}), where 
\begin{equation} \label{eq:Pauli noisy channel}
    K_i^\theta =\sqrt{p_i} V_\theta P_i U_\theta,\ i=1,\dots,r,
\end{equation}
$\sum_{i=1}^r p_i=1$ with each $p_i > 0$, and $\{P_i\}_i$ are ($n$-qubit) Pauli operators (having included the identity operator $I$ in $\{P_i\}_i$ \footnote{We assume that the Kraus operators of the Pauli noise channel contain a term proportional to identity, which is a common case in practice.}). $U_\theta$ and $V_\theta$ are both unitary operators encoding the parameter. The generators associated with the signal before and after the noise are $H^{(U)}:=iU_\theta^\dagger \dot{U}_\theta$ and $H^{(V)}:=iV_\theta^\dagger \dot{V}_\theta$, and we define $H_{\mathrm{tot}}:=H^{(U)}+H^{(V)}$. We decompose $H_{\mathrm{tot}}=\sum_k \alpha_kQ_k$, where $\{Q_k\}_k$ are ($n$-qubit) Pauli operators (including the identity), and each $\alpha_k \neq 0$ (see the proof in Ref.~\cite{SM_note}). 

\begin{theorem} \label{thm:HL Pauli noise}
    For Kraus operators given by Eq.~(\ref{eq:Pauli noisy channel}), if, at $\theta=\theta_0$, there exists some $Q_k$ such that (i) \footnote{$\langle a,b,\dots,c\rangle$ denotes the group generated by $a,b,\dots,c$. $\mathsf{Span}$ denotes the linear span of the set of group elements.}
    \begin{equation} \label{eq:Q not in span}
        Q_k \notin \mathsf{Span}\left\{\langle P_1, P_2, \dots, P_r\rangle\right\}
    \end{equation}
    and (ii) 
    \begin{equation} \label{eq:Q commutes}
        [Q_k, P_i] = 0,\ \forall i=1,\dots,r,
    \end{equation}
    then the output state QFI achieves the HL by applying the same unitary control $U_c=U_{\theta_0}^\dagger V_{\theta_0}^\dagger$ following each channel $T_\theta$. 
\end{theorem}

\emph{A working example}---Consider estimating a two-qubit noisy channel $T_\theta$ defined by Eq.~(\ref{eq:Liouville representation}), characterized by the Kraus operators $K_i^\theta=K_i^{(\mathrm{noise})}U_t(\theta)$,
where $U_t(\theta)=e^{-\mathrm i t H_0(\theta)}$ for the evolution time $t$, and the parameter $\theta$ of interest is the coupling strength in a Heisenberg model Hamiltonian 
\begin{equation}
    H_0(\theta)= \sigma_z\otimes I + I\otimes \sigma_z +\theta H_J
\end{equation}
for $H_J = \sigma_x\otimes \sigma_x+\sigma_y\otimes \sigma_y+\sigma_z\otimes \sigma_z$. The Pauli noise is described by
\begin{equation}
    \begin{aligned}
         K_1^{(\mathrm{noise})}&=\sqrt{1-p_1-p_2}I,\\
         K_2^{(\mathrm{noise})}&=\sqrt{p_1}\sigma_x\otimes \sigma_x, \\
         K_3^{(\mathrm{noise})}&=\sqrt{p_2}\sigma_x\otimes \sigma_y.
    \end{aligned}
\end{equation}
  
Note that not all $\left\{K_i^{\theta\dagger} K_j^\theta\right\}$ commute with each other, which goes beyond the requirement of the ancilla-free QEC code protocol in Ref.~\cite{Layden19PRL}. It is also impossible to find a codespace stabilized by the noise as in Ref.~\cite{Pereira2023}. Nevertheless, we can easily check that for this noise model the ancilla-free HL is achievable according to Theorem~\ref{thm:HL Pauli noise} \footnote{Note that we can choose $Q_k=\sigma_z\otimes \sigma_z$ in $H_{\mathrm{tot}}=\sigma_x\otimes\sigma_x+\sigma_y\otimes\sigma_y+\sigma_z\otimes\sigma_z$, and the linear span of group elements generated by the Pauli noise operators is $\mathsf{Span}\{I,\sigma_x\otimes\sigma_x,\sigma_x\otimes\sigma_y,I\otimes \sigma_z\}$.}.    

We remark that it is unnecessary to know that the channel to be estimated is in the form of Pauli noise \emph{a priori}. To apply Theorem~\ref{thm:HL with unitary control} and the algorithm for identifying the unitary control, it suffices to have the knowledge about such a ``tomographic'' description of its Liouville representation $T_\theta$ in the neighborhood of the true value $\theta_0$, bypassing the need to explicitly determine the type of the channel (in this case, the Pauli noise characterized by Kraus operators) and making it more convenient for dealing with the experimental data. 

Here, for simplicity of presentation, we assume that we know the forms of unitary $U_t(\theta)$ and the noise $\left\{K_i^{(\mathrm{noise})}\right\}$, respectively. Note that $T_\theta^\dagger T_\theta$ has four eigenvalues of $1$ with eigenvectors in the subspace $\mathcal P$ (associated with the projection $P$ on it) and 
\begin{equation}
    PT_\theta^\dagger \dot{T}_\theta P=-\mathrm i t P T_\theta^\dagger T_\theta [H_J \otimes I - I \otimes H_J^T]P
\end{equation}
has two nonzero eigenvalues, with mutually orthogonal eigenmatrices (up to multiplicative factors)
\begin{equation}
    \begin{aligned}
        R_1 &= U_t(\theta)^\dagger (\ketbra{01}{11}+\ketbra{10}{00}) U_t(\theta),\\
        R_2 &= R_1^{\dagger} = U_t(\theta)^\dagger(\ketbra{00}{10}+\ketbra{11}{01}) U_t(\theta).
    \end{aligned}
\end{equation}
We can choose the unitary control
\begin{equation} \label{eq:control example HL unitary control}
    U_c = U_t(\theta_0)^\dagger
\end{equation}
and take an input state, for example,
\begin{equation} \label{eq:input state example HL unitary control}
    \rho_0 = \frac{I}{2} \otimes \frac{I}{2} + \frac{\alpha}{2} (R_1 + R_1^{\dagger}) = U_t(\theta_0)^\dagger \left[\left(\frac{I}{2} + \alpha\sigma_x\right)\otimes \frac{I}{2}\right] U_t(\theta_0)
\end{equation}
for some $0<\alpha<1/2$ \footnote{Here we do not choose $\alpha=1/2$ corresponding to a pure state, such that the rank of the output state is independent of $\theta$ at $\theta=\theta_0$, to avoid the possible discontinuity issue of the QFI at the rank-changing point. See Refs.~\cite{Dominik17PRA,Seveso_2020,zhou2019exact,Ye22PRA} for more discussions.} and incorporate unitary control $U_c=U_t(\theta_0)^\dagger$ after each application of $T_\theta$. Note that if only the form of $T_\theta$ (instead of separating the unitary evolution and noise) is known, we can apply our algorithm to $T_\theta$ to find a numerical solution for $U_c$.

Remarkably, in this case, preparing the input state Eq.~(\ref{eq:input state example HL unitary control}) is fully robust to any local state preparation error on the second qubit, i.e., the second qubit can be arbitrarily initialized. The reason is that the regulated channel $(U_c\otimes U_c^*)T_{\theta}$ is unitarily diagonalizable at $\theta=\theta_0$. For an arbitrary state of the second qubit $\sigma = \frac{1}{2} (I + r_x\sigma_x + r_y\sigma_y+r_z\sigma_z)$ with $r_x^2+r_y^2+r_z^2=1$, all the contributions from $\sigma_x,\sigma_y,\sigma_z$ vanish in the asymptotic limit, as they are orthogonal to the peripheral eigenspace. Therefore, the choice of $\sigma$ has no detrimental effect on achieving the guaranteed HL.

In the End Matter, we conduct additional analysis to assess the error robustness of the state preparation and measurement (SPAM) as well as control in our strategy. The sequential nature of our approach is particularly well suited for addressing SPAM errors.

It is worth noting that control is essential for achieving the HL here. Without control, $T_\theta$ generically has only one fixed point state (the maximally mixed state) in the peripheral eigenspace, so no information can be retrieved after $N\rightarrow\infty$ applications in a control-free strategy. 

This informative example reveals three key distinctions between our approach and QEC: (i) We provide an algorithmic routine to find the probe state and identical unitary control for achieving the ancilla-free HL in an experiment-friendly scenario, while such a routine is lacking for constructing ancilla-free QEC codes for metrology. (ii) We do not use QEC recovery operations, which may require additional ancillae. (iii) The probe state by our approach does not necessarily lie in a QEC codespace that satisfies the Knill-Laflamme condition \cite{Knill97PRA}.
Overall, our formalism works in a more resource-deficient scenario, which can be of practical interest.

\emph{Discussion}---We establish a systematic formalism for ancilla-free quantum channel estimation. By vectorizing the quantum state and investigating the spectral properties of channels, we derive useful formulas for the QFI with ancilla-free sequential strategies and identify sufficient conditions for achieving the HL in both control-free and unitary-control-enhanced scenarios. This formalism is also applicable while allowing bounded ancillae, as the problem reduces to estimation of $\mathcal E_\theta \otimes \mathcal N_A$, where $\mathcal N_A$ is the noise channel acting on the ancillae. Moreover, our technique may also be used to investigate the metrological limits of ancilla-free parallel strategies, which have been partially explored with different methods previously \cite{Alipour14PRL,knysh2014true,Demkowicz14PRL,Zhou24PRAachieving}.

Our results have other surprising implications not covered in the main text. In Ref.~\cite{SM_note}, we identify examples demonstrating that it is sometimes possible to achieve the HL, even when the conventional HNKS condition is ill-defined, or in the absence of decoherence-free subspace \cite{Lidar98PRL}.

Our approach is less resource demanding compared to the conventional QEC approach and well suited for experimental design. Compared with our previous works \cite{Altherr21PRL,Liu23PRL,Liu24AQT}, the obtained protocol is not strictly optimal, but it has a performance guarantee of the precision scaling. The complexity of the algorithm is much lower and does not grow with $N$. The simple setup allows the circuit to be looped many times, thus facilitating experimental demonstration of quantum metrology for large $N$.

A very recent work \cite{zhou2024limits} proved that the HL is unattainable with unitary control for single-qubit noisy channel estimation when the noise strength is nonvanishing. Interestingly, our work shows that this conclusion does not extend to higher dimensions. 

\emph{Acknowledgments}---This work is supported by the National Natural Science Foundation of China via Excellent Young Scientists Fund (Hong Kong and Macau) Project No.~12322516, 
Ministry of Science and Technology, China (MOST2030) with Grant No.\ 2023200300600, Guangdong Provincial Quantum Science Strategic Initiative (No.~GDZX2303007 and No.~GDZX2203001), and the Hong Kong Research Grant Council through the Early Career Scheme Grant No.~27310822 and the General Research Fund Grant No.~17303923.  

\emph{Data availability}---The data that support the findings of this article are openly available \cite{github}.

\section*{End Matter}
\emph{On robustness against protocol imperfections}---Here, we provide more analysis for the example presented in the main text, discussing the robustness of our approach with respect to the SPAM and control errors. Notably, the robustness against SPAM errors arises from the sequential nature of our protocol and is a general feature not limited to the specific example discussed. More precisely, the probe state given by Theorem~\ref{thm:condition for HL} provides the flexibility to choose parameters $\alpha$ and $\beta$, making it robust to state preparation error. Similar arguments hold for robustness against measurement error.

We model the SPAM error $p_{\mathrm{SPAM}}$ as single-qubit depolarizing noise described by the channel $\mathcal N^{(\mathrm{dep})}(\rho) = (1-p_{\mathrm{SPAM}})\rho + (p_{\mathrm{SPAM}}/2)I$. Specifically, in the two-qubit example, $\mathcal N^{(\mathrm{dep})} \otimes \mathcal N^{(\mathrm{dep})}$ occurs at both the state preparation and measurement stages.

The control error is modeled as a classical Gaussian fluctuation in the control Hamiltonian, i.e., at each intervention step the actual control is $U_c = e^{-\mathrm i tH_c}$, where $H_c = -BH_0(\theta_0)$ and $B$ is a random variable drawn from the normal distribution with mean $\mu=1.0$ and variance $\sigma_{\mathrm{control}}^2$. We assume that $B$ is independent and identically distributed (i.i.d.) across different control steps.

We plot the normalized output QFI $F^Q(\rho_\theta)/N$ versus the number of queries $N$ in Fig.~\ref{fig:HL robustness}. We take $\theta_0 = 1.0$, $t=1.0$, and $p_1=p_2=0.1$. Ideally, without SPAM and control errors, we prepare the input state $\rho_0$ given by Eq.~(\ref{eq:input state example HL unitary control}) for $\alpha=0.499$ and choose the unitary control $U_c$ given by Eq.~(\ref{eq:control example HL unitary control}), which yields the HL as shown by the solid gray line. We also verify that in the ideal case, the input state of the second qubit can be arbitrarily chosen without affecting the QFI, as predicted by the theory.

\begin{figure}[!htbp]
    \centering
    \includegraphics[width=0.48\textwidth]{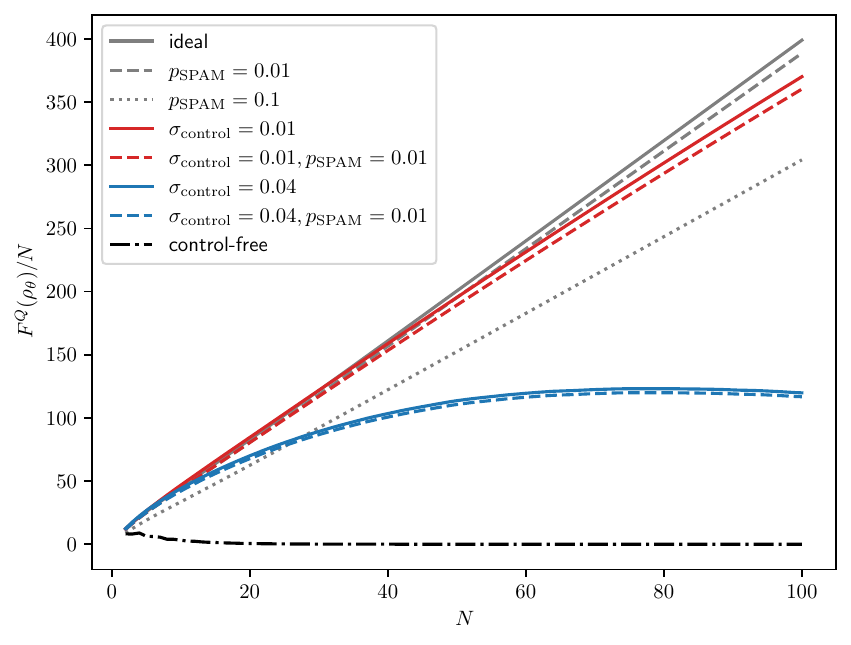}
    \caption{The normalized output QFI $F^Q(\rho_\theta)/N$ versus the number of queries $N$. The (expectation value of the) QFI is computed over $1000$ trials with random control errors. The gray lines represent different SPAM error rates ($p_{\mathrm{SPAM}}=0,0.01,0.1$) under ideal control ($\sigma_{\mathrm{control}}=0$). The red ($\sigma_{\mathrm{control}}=0.01$) and blue ($\sigma_{\mathrm{control}}=0.04$) lines show how the HL deviates as control error increases, with dashed lines corresponding to the presence of SPAM errors and solid lines indicating the absence of SPAM errors. The solid black line represents the performance of a control-free strategy.}
    \label{fig:HL robustness}
\end{figure}

In many typical experiments, the readout noise is dominant compared to gate noise \cite{Arute2019,Moses23PRX,Bluvstein2024,IBM_readout_note}. We find that the SPAM errors are not fatal for achieving the HL, indicated by dashed ($p_{\mathrm{SPAM}}=0.01$) and dotted ($p_{\mathrm{SPAM}}=0.1$) gray lines \footnote{The chosen values of $p_{\mathrm{SPAM}}$ are common in current quantum devices \cite{IBM_readout_note}}. This can be expected, since our sequential strategy repeatedly amplifies the signal and the SPAM noise has only a constant effect in the asymptotic limit.

Control errors are harder to tackle since they accumulate as $N$ grows. Nevertheless, when the classical error in control is small ($\sigma_{\mathrm{control}}=0.01$), we can still observe a QFI scaling close to the HL. The deviation from the HL becomes more obvious for larger fluctuations ($\sigma_{\mathrm{control}}=0.04$). 

As a benchmark, the dash-dotted black line depicts the performance of a control-free strategy in the ideal case without SPAM and control errors. The QFI vanishes for large $N$, as $T_{\theta_0}$ is a mixing channel, transforming a generic input state into the fixed point state.

Finally, we provide a simple analysis of the small perturbations in control. Let us denote the ideal control by $U_c^{(\mathrm{ideal})}=e^{\mathrm i t H_0(\theta_0)}$, and consider the actual control $U_c=e^{-\mathrm i t H_c}$ with perturbed Hamiltonian $H_c=-BH_0(\theta_0)$. Assume $B=1+\epsilon$ for small error $\epsilon$ (where $\epsilon$ and $\sigma_{\mathrm{control}}$ are within the same order of magnitude). We then have
\begin{equation} \label{eq:control error norm}
    \norm{U_c-U_c^{(\mathrm{ideal})}} \le \epsilon t \norm{H_0(\theta_0)} + O(\epsilon^2),
\end{equation}
where $\norm{\cdot}$ denotes the operator norm. Without control error, the ideal output state is $\kett{\rho_\theta^{(\mathrm{ideal})}} = \left(\hat T_\theta^{(\mathrm{ideal})}\right)^N \kett{\rho_0}$, where $\hat T_\theta^{(\mathrm{ideal})} := \left(U_c^{(\mathrm{ideal})} \otimes U_c^{(\mathrm{ideal})*}\right)T_\theta$. Similarly, the actual output state is $\kett{\rho_\theta} = \hat T_\theta^N \kett{\rho_0}$ for $\hat T_\theta := (U_c\otimes U_c^*)T_\theta$. From Eq.~(\ref{eq:control error norm}) we obtain (to the first order approximation)
\begin{equation}
    \begin{aligned}
        \norm{\hat T_\theta - \hat T_\theta^{(\mathrm{ideal})}} &\le 2\epsilon t \norm{H_0(\theta_0)}, \\
        \norm{\dot{\hat T}_\theta - \dot{\hat T}_\theta^{(\mathrm{ideal})}} &\le 4\epsilon t^2 \norm{H_0(\theta_0)}\norm{H_J},
    \end{aligned}
\end{equation}
having used $\norm{T_\theta}= 1$ for this noise model. By defining $\delta_1 := 2\epsilon t \norm{H_0(\theta_0)}$ and $\delta_2:=4\epsilon t^2 \norm{H_0(\theta_0)}\norm{H_J}$, we can bound the closeness of the output state to the ideal one:
\begin{equation}
    \begin{aligned}
        \norm{\rho_\theta-\rho_\theta^{(\mathrm{ideal})}}_F &\le N\delta_1 \norm{\rho_0}_F, \\
        \norm{\dot \rho_\theta-\dot \rho_\theta^{(\mathrm{ideal})}}_F &\le [N\delta_2 + 2N(N-1)\delta_1t\norm{H_J}]\norm{\rho_0}_F,
    \end{aligned}
\end{equation}
where $\norm{\cdot}_F$ denotes the Frobenius norm. By the continuity of QFI \cite[Theorem~1]{Rezakhani19PRA}, we can bound the relative deviation of the QFI from the ideal case:
\begin{equation}
    \begin{aligned}
        &\mathrel{\phantom{\le}} \frac{1}{N^2}\left\lvert F_Q(\rho_\theta)-F_Q\left(\rho_\theta^{(\mathrm{ideal})}\right)\right\rvert \\
        &\le K_1 \norm{\rho_\theta-\rho_\theta^{(\mathrm{ideal})}}_1+ \frac{K_2}{N} \norm{\dot \rho_\theta-\dot \rho_\theta^{(\mathrm{ideal})}}_1 \\
        &\le K\sqrt{d} \norm{\rho_0}_F\{[N+2(N-1)t\norm{H_J}]\delta_1 + \delta_2\}
    \end{aligned}
\end{equation}
for constants $K_1$ and $K_2$ and $K:=\max\{K_1,K_2\}$ (see Ref.~\cite{Rezakhani19PRA} for the detailed definitions of these constants; note that here they are bounded as $N$ grows). Here, the trace norm $\norm{M}_1 \le \sqrt{d}\norm{M}_F$ for a $d\times d$ matrix $M$. 
Therefore, a sufficient condition for the quadratic (i.e., ``HL"-like) dependence of the QFI on the number of queries $N$, as shown in Fig.~\ref{fig:HL robustness}, is that $N$ is small compared with $1/\epsilon$.

\providecommand{\noopsort}[1]{}\providecommand{\singleletter}[1]{#1}%
%


\onecolumngrid
\section*{Supplemental Material}

\appendix

Section \ref{sec:preliminaries} introduces the preliminaries on the vectorization technique and the spectral properties of quantum channels, which turn out to be useful throughout the proof. Section \ref{sec:proof lemma associated QFI} contains the proof of Lemma \ref{lem:associated QFI} for a closed-form expression of the associated QFI. Section \ref{sec:proof lemma asymptotic associated QFI} contains the proof of Lemma \ref{lem:asymptotic associated QFI} for connecting the output associated QFI to the spectral properties of a quantum channel in the asymptotic limit. In Sections \ref{sec:proof theorem condition for HL} and \ref{sec:proof theorem HL with unitary control} we prove the conditions for achieving the ancilla-free HL without control and with unitary control, respectively summarized in Theorems~\ref{thm:condition for HL} and \ref{thm:HL with unitary control}. A detailed description and proof of the algorithm for unitary control is given in Section \ref{sec:alg for control}. In Section \ref{sec:proof HL Pauli noise} we prove conditions for achieving the HL in parameter estimation under Pauli noise. Finally, Section \ref{sec:supplemental examples} presents supplemental examples demonstrating the singularity of HNKS condition and the attainablity of the HL in absence of DFS.

\section{Preliminaries} \label{sec:preliminaries}
\subsection{Vectorization}
For a linear operator $A$ on a $d$-dimensional Hilbert space $\mathcal H$, we denote its vectorization as 
\begin{equation}
    \kett{A} := \sum_{ij}\mel{i}{A}{j}\ket{i}\ket{j},
\end{equation}
where $\{\ket{i}\in\mathcal H\}_i$ and $\{\ket{j}\in\mathcal H\}_j$ represent the computational basis. It is easy to check that 
\begin{equation}
    \kett{A}=(A\otimes I) \ket{\Phi}
\end{equation}
for the (unnormalized) maximally entangled state $\ket{\Phi}=\sum_i\ket{i}\ket{i}$. We denote the adjoint of $\kett{A}$ as $\bbra{A}:=\kett{A}^\dagger$ and the complex conjugate of $\kett{A}$ as $\kett{A^*}$.

We present several properties that are straightforward to obtain. Given an operator $M=\sum_k\ketbra{u_k}{v_k}$ for all $\ket{u_k},\ket{v_k} \in \mathcal H$, we have
\begin{equation}
    \kett{M} = \sum_k \ket{u_k} \otimes \ket{v_k^*}.
\end{equation}

The Hilbert-Schmidt inner product between two operators $A$ and $B$ is preserved by vectorization, i.e., 
\begin{equation}
    \Tr\left(A^\dagger B\right) = \bbrakett{A}{B}.
\end{equation}

A particularly useful identity we will use throughout the proofs is 
\begin{equation} \label{eq:identity in vec}
    \kett{ABC} = \left(A \otimes C^T\right) \kett{B},
\end{equation}
for any operators $A$, $B$ and $C$.

Interested readers can refer to Refs.~\cite{watrous2018theory,Alipour14PRL} for more details pertaining to the operator-vector correspondence.

\subsection{Spectral properties of quantum channels}
By vectorizing the density operator $\rho$ of a $d$-dimensional quantum state into $\kett{\rho}$, the action of a quantum channel $\mathcal E (\rho) = \sum_{i=1}^r K_i \rho K_i^\dagger$ can be expressed as
\begin{equation}
    \kett{\mathcal E (\rho)} = \sum_{i=1}^r\left(K_i\otimes K_i^*\right)\kett{\rho},
\end{equation}
having used the identity Eq.~(\ref{eq:identity in vec}). Hence, we can define a $d^2 \times d^2$ transition matrix
\begin{equation}
    T:=\sum_{i=1}^r K_i \otimes K_i^*.
\end{equation}
to represent the quantum channel, and its action is given by $\kett{\mathcal E(\rho)}=T\kett{\rho}$.

As $T$ is associated with a CPTP map, it has many useful properties \cite{wolf2012quantum,Burgarth_2013,watrous2018theory}, and we will review some relevant ones here. Specifically, we consider quantum channels with equal input and output spaces, and assign a spectrum $\{\lambda_i\}_i$ to $T$ such that $T\kett{R_i}=\lambda_i \kett{R_i}$. For clarity we call $\{\kett{R_i}\}_i$ the eigenvectors and $\{R_i\}_i$ the eigenmatrices. Note that $T$ is not necessarily diagonalizable and admits up to $d^2$ eigenvalues.
\begin{enumerate}
    \item The eigenvalues of $T$ are either real or come in complex conjugate pairs, i.e., if $\lambda_i$ is an eigenvalue of $T$, then $\lambda_i^*$ is also an eigenvalue of $T$. In fact, we have $T\kett{R_i}=\lambda_i \kett{R_i}$ and $T\kett{R_i^\dagger}=\lambda_i^* \kett{R_i^\dagger}$. If $\lambda_i$ is real, we can take $R_i$ to be Hermitian.
    \item All the eigenvalues of $T$ satisfy $|\lambda_i| \le 1$, confined in the unit circle on the complex plane. Those eigenvalues satisfying $|\lambda_i|=1$ are called \emph{peripheral eigenvalues}. We denote the set of indices of peripheral eigenvalues by $\mathsf S:=\{i \mid |\lambda_i| = 1\}$.
    \item All Jordan blocks for peripheral eigenvalues are one-dimensional \cite[Proposition 6.2]{wolf2012quantum}.
    \item Any quantum channel must have at least one \emph{fixed point state}, i.e., there exists a density operator $\rho_*$ such that $T\kett{\rho_*} = \kett{\rho_*}$.
    \item A quantum channel is said to be \emph{ergodic} if there exists a unique fixed point state $\rho_*$. This requirement is strong enough to guarantee that $T$ has only one \emph{fixed point}, i.e., eigenvector associated with unit eigenvalue $\lambda=1$.
    \item A quantum channel $\mathcal E$ is said to be \emph{mixing} if there exists a fixed point state $\rho_*$ such that
    \begin{equation}
        \lim_{N\rightarrow \infty} \lVert \mathcal E^N(\rho) - \rho_*\rVert = 0,\ \forall \rho\ \text{as a density operator}.
    \end{equation}
    A channel is mixing iff it is ergodic and does not have any other peripheral eigenvalue except for $\lambda=1$.
    \item As $T$ is a trace-preserving map, we have $\Tr R_j = 0$ if $\lambda_j \neq 1$.
\end{enumerate}

\emph{Remark}---Our framework establishes a close relation between the metrological limit and the spectral properties of quantum
channels. It thus provides a new perspective on quantum metrology drawn from the asymptotic theory of quantum
channels, where researchers have shown great interest in quantum ergodicity and mixing \cite{Evans1978Spectral,wolf2012quantum,Burgarth_2013,Albert2019asymptoticsof,Neill2016,Movassagh21PRX}. We expect that our results may also have applications in these areas beyond quantum metrology.

\section{Proof of Lemma \ref{lem:associated QFI}} \label{sec:proof lemma associated QFI}
\begin{proof}
    By noting that $\tilde \rho_\theta^2=\tilde \rho_\theta$, we can take the SLD $L_{\tilde \rho_\theta}=2\dot{\tilde{\rho}}_\theta$. It follows that    
    \begin{equation}
        \begin{aligned}
            L_{\tilde \rho_\theta} &= 2\left\{\frac{\kett{\dot \rho_\theta}\bbra{\rho_\theta}}{\Tr\left(\rho_\theta^2\right)} + \frac{\kett{\rho_\theta}\bbra{\dot \rho_\theta}}{\Tr\left(\rho_\theta^2\right)} - \frac{d\Tr\left(\rho_\theta^2\right)/d\theta}{\left[\Tr\left(\rho_\theta^2\right)\right]^2} \kett{\rho_\theta}\bbra{\rho_\theta}\right\}\\
            &=\frac{2}{\Tr\left(\rho_\theta^2\right)}\left(\kett{\dot \rho_\theta}\bbra{\rho_\theta} + \kett{\rho_\theta}\bbra{\dot \rho_\theta} - 2\bbrakett{\rho_\theta}{\dot\rho_\theta} \tilde \rho_\theta\right),
        \end{aligned}
    \end{equation}
    having used $d\Tr\left(\rho_\theta^2\right)/d\theta=2\Tr\left(\rho_\theta \dot \rho_\theta\right)=2\bbrakett{\rho_\theta}{\dot\rho_\theta}$. Then we have
    \begin{equation}
            L_{\tilde \rho_\theta}^2 = \frac{4}{\left[\Tr\left(\rho_\theta^2\right)\right]^2}\left[\Tr\left(\rho_\theta^2\right)\bbrakett{\dot \rho_\theta}{\dot \rho_\theta}\tilde \rho_\theta - \bbrakett{\rho_\theta}{\dot \rho_\theta}\left(\kett{\dot \rho_\theta}\bbra{\rho_\theta} + \kett{\rho_\theta}\bbra{\dot \rho_\theta}\right) + \Tr\left(\rho_\theta^2\right)\kett{\dot \rho_\theta}\bbra{\dot \rho_\theta}\right],
    \end{equation}
    which yields
    \begin{equation}
            \tilde F^Q \left(\tilde \rho_\theta\right) = \Tr\left(\tilde \rho_\theta L_{\tilde \rho_\theta}^2\right) =  4 \left\{\frac{\bbrakett{\dot \rho_\theta}{\dot \rho_\theta}}{\Tr\left(\rho_\theta^2\right)} - \left[\frac{\bbrakett{\rho_\theta}{\dot \rho_\theta}}{\Tr\left(\rho_\theta^2\right)}\right]^2\right\}.
    \end{equation}
\end{proof}

\section{Proof of Lemma \ref{lem:asymptotic associated QFI}} \label{sec:proof lemma asymptotic associated QFI}
\begin{proof}
    In general $T_\theta$ may or may not be diagonalizable. If $T_\theta$ is not diagonalizable, we can add some additional linearly independent vectors $\kett{R_i}$ orthogonal to the subspace of peripheral eigenvectors, such that $\{\kett{R_i}\}_{i=1}^{d^2}$ forms a Jordan basis which spans the whole $d^2$-dimensional vector space. Under such choice we have the decomposition $\kett{\rho_0}=\sum_{i=1}^{d^2}a_{i}\kett{R_i}$. Without loss of generality, we can always choose $a_i$ to be nonnegative real numbers and $\bbrakett{R_i}{R_i}=1$ for all $i$. By our choice $R_1=\frac{1}{\sqrt{\Tr\left(\rho_*^2\right)}}\rho_*$ for a fixed point state $\rho_*$, and $a_1=1/\Tr R_1$. Then we have 
    \begin{equation}
        \kett{\rho_\theta}=T_\theta^N\kett{\rho_0}=\sum_{i\in \mathsf{S}}a_i\lambda_i^N\kett{R_i},\ N\rightarrow\infty,
    \end{equation}
    since the $N$-th power of any Jordan block corresponding to $|\lambda_i|<1$ converges to 0 as $N\rightarrow \infty$. It follows that
    \begin{equation} \label{eq:derivative rho theta}
        \kett{\dot \rho_\theta}=\sum_{i\in \mathsf{S}}\left(N a_i\lambda_i^{N-1}\dot \lambda_i \kett{R_i} + \dot a_i\lambda_i^N\kett{R_i} + a_i\lambda_i^N\kett{\dot R_i}\right).
    \end{equation}
    Thus we obtain
    \begin{equation}
        \begin{aligned}
            \bbrakett{\dot \rho_\theta}{\dot \rho_\theta} &= N^2 \sum_{i,j\in\mathsf S}(\lambda_i^*\lambda_j)^{N-1} a_ia_j\dot{\lambda}_i^*\dot{\lambda}_j\bbrakett{R_i}{R_j} \\
            &+ 2N \sum_{i,j\in\mathsf S}(\lambda_i^*\lambda_j)^{N-1}\lambda_i^*\dot{\lambda}_j\left(\dot a_ia_j\bbrakett{R_i}{R_j} + a_ia_j\bbrakett{\dot R_i}{R_j}\right) \\
            &+ \sum_{i,j\in\mathsf S}(\lambda_i^*\lambda_j)^N\left(\dot a_i \dot a_j\bbrakett{R_i}{R_j} + a_ia_j\bbrakett{\dot R_i}{\dot R_j} + \dot a_ia_j\bbrakett{R_i}{\dot R_j} + a_i\dot a_j\bbrakett{\dot R_i}{R_j} \right) ,
        \end{aligned}
    \end{equation}
    and
    \begin{equation}
        \begin{aligned}
            \bbrakett{\rho_\theta}{\dot \rho_\theta} &= N \sum_{i,j\in\mathsf S}(\lambda_i^*\lambda_j)^{N-1}a_ia_j\lambda_i^*\dot{\lambda}_j\bbrakett{R_i}{R_j} + \sum_{i,j\in\mathsf S}(\lambda_i^*\lambda_j)^N\left(a_i\dot a_j\bbrakett{R_i}{R_j} + a_ia_j\bbrakett{R_i}{\dot R_j}\right),
        \end{aligned}
    \end{equation}
    having used the property that eigenvalues come in complex conjugate pairs. By using Lemma \ref{lem:associated QFI} and the relation $\Tr\left(\rho_\theta^2\right)=\bbrakett{\rho_\theta}{\rho_\theta}=\sum_{i,j\in\mathsf S}(\lambda_i^*\lambda_j)^Na_ia_j\bbrakett{R_i}{R_j}$, we then obtain Eq.~(\ref{eq:asymptotic associated QFI}).
\end{proof}

\section{Proof of Theorem~\ref{thm:condition for HL}} \label{sec:proof theorem condition for HL}
\begin{proof}
Define 
\begin{equation}
    \kett{D_\theta}:=\sum_{i\in \mathsf{S}}N a_i\lambda_i^{N-1}\dot \lambda_i \kett{R_i},
\end{equation}
which constitutes a part of the contribution to $\kett{\dot \rho_\theta}$ in Eq.~(\ref{eq:derivative rho theta}). Then by Lemma \ref{lem:asymptotic associated QFI} it follows that
\begin{equation}
    \lim_{N\rightarrow\infty}\frac{\tilde F^Q \left(\tilde \rho_\theta\right)}{N^2} =  \frac{4}{N^2} \left\{\frac{\bbrakett{D_\theta}{D_\theta}}{\Tr\left(\rho_\theta^2\right)} - \left[\frac{\bbrakett{\rho_\theta}{D_\theta}}{\Tr\left(\rho_\theta^2\right)}\right]^2\right\}.
\end{equation}
By the Cauchy–Schwarz inequality and noting that the QFI is a real number, we have
\begin{equation}
    \bbrakett{\rho_\theta}{D_\theta}^2 \le \bbrakett{\rho_\theta}{\rho_\theta} \bbrakett{D_\theta}{D_\theta} = \Tr(\rho_\theta^2)\bbrakett{D_\theta}{D_\theta},
\end{equation}
where the equality holds if and only if $\kett{D_\theta}\propto\kett{\rho_\theta}$. Equivalently this condition can be formulated as $\dot \lambda_i/\lambda_i = \dot \lambda_j/\lambda_j$
for all $i,j\in \mathsf S$ associated with $a_i \ne 0$ and $a_j \ne 0$. When the condition is violated, we can obtain $\lim_{N\rightarrow\infty}\frac{\tilde F^Q \left(\tilde \rho_\theta\right)}{N^2} >0$ that yields the HL.

Now we denote the fixed point state by $\rho_*$ with the eigenvalue $\lambda_1 = 1$. If there exists some $\lambda_j$ such that $|\lambda_j| = 1$ and $\dot \lambda_j\neq 0$, we will see that it is always possible to choose an input state with $a_{1} \neq 0$ and $a_{j} \neq 0$. By the trace-preserving property of the quantum channel we have $\lim_{\theta\rightarrow\theta_0}\Tr R_j=0$, so we can take an input state $\rho_0=(1-\alpha)I/d + \alpha \rho_* +\beta\left(R_j + R_j^\dagger\right) $ for some $\alpha >0$ and $\beta >0$. Note that if $R_j$ is proportional to a Hermitian matrix, we have the freedom to simply take $R_j$ to be Hermitian. Since $\dot \lambda_1=0$ and $\dot \lambda_j \neq 0$, this choice of the input state allows for the HL.
\end{proof}

\section{Proof of Theorem~\ref{thm:HL with unitary control}} \label{sec:proof theorem HL with unitary control}
\begin{proof}
    Our goal is to show that the regulated channel $(U_c\otimes U_c^*) T_\theta$ has a peripheral eigenmatrix $R_0$ and the associated eigenvalue has a nonzero derivative in the limit of $\theta\rightarrow\theta_0$, even if this may not hold for $T_\theta$ itself.
    
    First we will show that $\kett{R_0}\in\mathcal P$ is the limit of a peripheral eigenvector of $(U_c \otimes U_c^*)T_\theta$. Denoting by $\mu_0$ the eigenvalue of $T_\theta^\dagger\dot T_\theta$ associated with $R_0$, we have
    \begin{equation}
        \mu_0\kett{R_0} = \lim_{\theta\rightarrow\theta_0}T_{\theta}^\dagger \dot T_\theta \kett{R_0}= T_{\theta_0}^\dagger \lim_{d\theta\rightarrow 0} \frac{T_{\theta_0+d\theta}-T_{\theta_0}}{d\theta}\kett{R_0}=T_{\theta_0}^\dagger(U_c^\dagger \otimes U_c^T)(U_c \otimes U_c^*) \lim_{d\theta\rightarrow 0} \frac{T_{\theta_0+d\theta}-T_{\theta_0}}{d\theta}\kett{R_0},
    \end{equation}
    so for $d\theta\rightarrow 0$ we obtain $(U_c\otimes U_c^*)T_{\theta_0+d\theta} \kett{R_0} \propto \kett{R_0}$, by noting that $T_{\theta_0}^\dagger (U_c^\dagger \otimes U_c^T)\kett{R_0} = \kett{R_0}$ and $(U_c \otimes U_c^*)T_{\theta_0}\kett{R_0} = \kett{R_0}$. Indeed, $R_0$ is the limit of a peripheral eigenmatrix of $(U_c\otimes U_c^*) T_{\theta}$. 
    
    Then we will show that for any peripheral eigenvalue $\lambda_j$ of $(U_c\otimes U_c^*) T_\theta$ with the normalized eigenvector $\kett{R_j}$, we have $\lambda_j^*\dot \lambda_j = \bbra{R_j} T_\theta^\dagger \dot T_\theta \kett{R_j}$. By taking derivatives on both sides of $(U_c\otimes U_c^*)T_\theta\kett{R_j}=\lambda_j\kett{R_j}$ we obtain
    \begin{equation}
        (U_c\otimes U_c^*)\dot T_\theta\kett{R_j} + (U_c\otimes U_c^*)T_\theta \kett{\dot R_j} =\dot \lambda_j \kett{R_j} + \lambda_j \kett{\dot R_j}.
    \end{equation}
    It follows that
    \begin{equation}
        \bbra{R_j} T_\theta^\dagger \dot T_\theta\kett{R_j} + \bbra{R_j} T_\theta^\dagger T_\theta \kett{\dot R_j} =\bbra{R_j} \lambda_j^* \dot \lambda_j \kett{R_j} + \bbra{R_j}\lambda_j^* \lambda_j \kett{\dot R_j}.
    \end{equation}
    Note that $\bbra{R_j}T_\theta^\dagger T_\theta = \bbra{R_j}$ since $\kett{R_j} \in \mathcal P$. Then we obtain $\lambda_j^*\dot \lambda_j = \bbra{R_j} T_\theta^\dagger \dot T_\theta \kett{R_j}$.

    $P T_\theta^\dagger\dot T_\theta P\neq 0$ thus implies the existence of some $\dot \lambda_j \neq 0$. In particular, $R_0$ is such an eigenmatrix of $P T_\theta^\dagger\dot T_\theta P$ with eigenvalue $\mu_0=\lambda_0^*\dot \lambda_0$. Finally, by Theorem~\ref{thm:condition for HL} the repeated application of $(U_c\otimes U_c^*) T_\theta$ can achieve the HL.
\end{proof}

\section{Algorithm for identification of the unitary control} 
\label{sec:alg for control}
In Algorithm~\ref{alg:unitary control}, we delineate in detail the procedure for identifying the unitary control in Theorem~\ref{thm:HL with unitary control} for achieving the HL. The code implementation is openly available on GitHub \cite{github}.

\SetKwComment{Comment}{/* }{ */}
\begin{algorithm}[!htbp] \label{alg:unitary control}
\caption{Identification of a unitary control $U_c$ in Theorem~\ref{thm:HL with unitary control}.}
    \DontPrintSemicolon
    \KwIn{$R_0$ in Eq.~(\ref{eq:unitary equivalence})} 
    \KwOut{Unitary control $U_c$}
    $R_1 \gets R_0+R_0^\dagger,\ R_2 \gets \mathrm i(R_0-R_0^\dagger)$ \\
    $\mathrm{List_{in}} \gets [R_1,R_2],\ \mathrm{List_{out}} \gets [\sum_kK_k^{\theta_0}R_1K_k^{\theta_0\dagger},\sum_kK_k^{\theta_0}R_2K_k^{\theta_0\dagger}]$ \\
    $\mathrm{Vectors_{in-out}}\gets[]$ \Comment*[r]{To store two sets of vectors connected by a unitary transformation $U_c^\dagger$}
    \For{$\mathrm{List}\ \mathbf{in}\ [\mathrm{List_{in}}, \mathrm{List_{out}}]$}{
    \For{$j\gets 1$ \KwTo $2$}
    {$R\gets \mathrm{List}[j]$ \\
    Spectral decomposition $R=\sum_{k=1}^{n_\Pi} \alpha_k \Pi_k$, where $\Pi_k=\sum_{a=1}^{n_v^{(k)}}\dyad{v_{ka}}$  \Comment*[r]{$\{\Pi_k\}_k$ is a set of projections; $\alpha_k$ are different for different $k$ and follow the descending order}
    $S_k \gets \left[\ket{v_{k1}},\dots,\ket{v_{kn_v^{(k)}}}\right],\ \forall k=1,\dots,n_\Pi$ \Comment*[r]{$S_k$ is a subspace spanned by orthonormal $\{\ket{v_{ka}}\}_a$}
    $\mathcal S_j\gets[S_1,\dots,S_{n_\Pi}]$ \Comment*[r]{$\mathcal S_j$ is a list of subspaces corresponding to $R$}
    }
    $\mathcal S \gets \bigcup_{j=1}^m\mathcal S_j$ \Comment*[r]{$\mathcal S$ is the union of all subspaces}
    \Repeat{$\mathrm{for\ every\ pair\ of\ subspaces}\ S_a, S_b\ \mathrm{in}\ \mathcal S\ \mathrm{with\ projections}\ \Pi_a,\Pi_b,\ \mathrm{either}\ \Pi_a\Pi_b=0,\ \mathrm{or}\ \Pi_a\ \mathrm{and}\ \Pi_b\ \mathrm{have\ the\ same\ rank}\ \mathrm{and\ all\ nonzero\ singular\ values\ of}\ \Pi_a\Pi_b\ \mathrm{are\ the\ same}$}
    {
    \For{$\mathrm{every\ pair\ of}\ S_x, S_y\ \mathbf{in}\ \mathcal S$}
    {$\Pi_x,\Pi_y$ are the projections on $S_x, S_y$\\
    Compute the SVD of $\Pi_x\Pi_y=UDV^\dagger$\\
    Obtain a list of subspaces $[X_1,\dots,X_{n_s}]$, each spanned by the column vectors of $U$ corresponding to the same \emph{nonzero} singular value in the ascending order \\
    Obtain a list of subspaces $[Y_1,\dots,Y_{n_s-1}]$, each spanned by the row vectors of $V^T$ corresponding to the same \emph{nonzero and nonunity} singular value in the ascending order \\
    $X_{{n_s+1}}$ is the orthogonal complement to $X_1,\dots,X_{n_s}$ in $S_x$ and $Y_{{n_s}}$ is the orthogonal complement to $Y_1,\dots,Y_{n_s-1}$ in $S_y$ \\
    $\mathcal S_{xy}\gets[X_1,\dots,X_{n_s+1},Y_1,\dots,Y_{n_s}]$ \Comment*[r]{The angles between $S_x,S_y$ are uniquely mapped to angles between subspaces in $\mathcal S_{xy}$}
    }
    $\mathcal S \gets \bigcup_{x,y} \mathcal S_{xy}$ \\
    Remove redundant subspaces in $\mathcal S$ if they are associated with identical projections
    }
    $\mathcal S' \gets []$ \Comment*[r]{$\mathcal S'$ is to store subspaces with basis vectors in a uniquely determined canonical order}
    \Repeat{$\mathcal S\ \mathrm{is\ empty}$}
    {Add the first subspace $S_1$ in $\mathcal S$ to $\mathcal S'$ and remove $S_1$ from $\mathcal S$ \\
    Find all the remaining subspaces $Q_1,\dots,Q_L$ in $\mathcal S$ which are not orthogonal to $S_1$, and remove them from $\mathcal S$\\
    \For{$j\gets 1$ \KwTo $L$}
    {$\Pi_1,\Gamma_j$ are rank-$n_j$ projections on $S_1,Q_j$\\
    Construct isometry $\tilde U$ with each column as a basis vector in subspace $S_1$\\
    $\tilde{V}^\dagger \gets \tilde{D}^{-1}\tilde U^\dagger\Pi_1\Gamma_j$, where $\tilde{D}$ is diagonal with all diagonal elements as the nonzero singular value of $\Pi_1\Gamma_j$\\  
    $Q_j \gets [\ket{v_1},\dots,\ket{v_{n_j}}]$, where $\ket{v_1},\dots, \ket{v_{n_j}}$ are the row vectors of $\tilde{V}^T$\\
    Add $Q_j$ to $\mathcal S'$ and remove $Q_j$ from $\mathcal S$
    }
    }
    $\mathrm{List_{vectors}} \gets []$\\
    Add all the basis vectors of all the subspaces in $\mathcal S'$ to $\mathrm{List_{vectors}}$, following the order established so far\\
    Add $\mathrm{List_{vectors}}$ to $\mathrm{Vectors_{in-out}}$
    }
    For $i=1,2$, construct matrices $M_i$, with each column as a basis vector in $\mathrm{Vectors_{in-out}}[i]$ following the established order \\
    Compute the SVD of $M_2M_1^\dagger = U'D'V^{\prime\dagger}$\\
    $U_c\gets V'U^{\prime\dagger}$ \Comment*[r]{$U_c^\dagger$ maps vectors in $\mathrm{Vectors_{in-out}}[1]$ to vectors in $\mathrm{Vectors_{in-out}}[2]$}
    \textbf{Sanity check:} If Eq.~(\ref{eq:unitary equivalence}) is satisfied by $U_c$, output $U_c$ and succeed; otherwise, there does not exist a unitary satisfying Eq.~(\ref{eq:unitary equivalence}).
\end{algorithm}
\begin{proof}
    First, it is easy to see that the existence of a unitary transformation between two sets of Hermitian matrices is equivalent to the existence of a unitary transformation between two sets of subspaces, if we apply the spectral decomposition to each Hermitian matrix. Denote by $\{P_1,\dots,P_K\}$ the subspaces generated by the spectral decompostion of all matrices in $\{R_i\}_{i=1}^2$, and by $\{Q_1,\dots,Q_K\}$ the subspaces generated by the spectral decompostion of all matrices in $\{\sum_kK_k^{\theta_0}R_iK_k^{\theta_0\dagger}\}_{i=1}^2$.
    
    We start from the characterization of the principal angles between two subspaces $F$ and $G$. It is always possible to choose a set of orthonormal basis vectors $V_F=\{\ket{u_1},\dots,\ket{u_\alpha},\ket{v_1},\dots,\ket{v_\beta},\ket{\psi_1},\dots,\ket{\psi_\gamma}\}$ for $F$ and a set of orthonormal basis vectors $V_G=\{\ket{u_1'},\dots,\ket{u_\alpha'},\ket{w_1},\dots,\ket{w_\beta},\ket{\phi_1},\dots,\ket{\phi_\delta}\}$ for $G$, such that any vector in the set $V_F$ and any vector in the set $V_G$ are othogonal, except for 
    \begin{equation}
        \braket{v_i}{w_i} = \cos(\eta_{i}),\ \forall i=1,\dots,\beta
    \end{equation}
    and
    \begin{equation}
        \braket{u_i}{u_i'} = 1,\ \forall i=1,\dots,\alpha,
    \end{equation}
where $0<\eta_i<\pi/2$ are nontrivial principal angles. We say this is a canonical order of basis vectors for these two subspaces $F$ and $G$. These principal angles can be determined by performing the SVD of $\Pi_F\Pi_G$, as $\cos(\eta_j)$ is the singular value. By the uniqueness of SVD, $\{\ket{u_i}=\ket{u_i'}\}_{i=1}^\alpha$ are unique up to a unitary transformation, $\{\ket{v_j}\}$ and $\{\ket{w_j}\}$ associated with the same $\eta_j$ are unique up to a simultaneous unitary transformation, $\{\psi_i\}_{i=1}^\gamma$ are unique up to a unitary transformation, and $\{\phi_i\}_{i=1}^\delta$ are unique up to a unitary transformation. The geometrical relation between two subspaces $F$ and $G$ are thus fully characterized by the angles between a set of lower-dimensional subspaces.

Iteratively, we repeat the procedure of reducing the dimensions of subspaces, for all pairs of subspaces in the set $\{P_1,\dots,P_K\}$, until any pair of the dimension-reduced subspaces are either mutually orthogonal or have the same dimensions and identical principal angles. Denote by $\{\tilde{P}_1,\dots,\tilde{P}_M\}$ this set of subspaces iteratively obtained so far. We then fix a canonical order of basis vectors for all these subspaces through such a procedure: we first choose an order of basis vectors for subspace $\tilde{P}_1$ with the associated projection $\tilde{\Pi}_1$. Then for any subspace $\tilde{P}_j$ (with the projection $\tilde{\Pi}_j$) in $\{\tilde{P}_2,\dots,\tilde{P}_M\}$ which is nonorthogonal to $\tilde{P}_1$, the canonical order of basis vectors is uniquely determined by fixing the isometry $\tilde U$ in the SVD of $\tilde{P}_1\tilde{\Pi}_j=\tilde{U}\tilde{D}\tilde{V}^\dagger$, where the diagonal elements of $\tilde{D}$ only contain nonzero singular values. We can thus choose $\tilde{V}^\dagger=\tilde{D}^{-1}\tilde{U}^\dagger\tilde{P}_1\tilde{\Pi}_j$, where each column of $\tilde U$ is the predetermined canonical basis vector for subspace $\tilde{\Pi}_1$. After finding all subspaces $P_j$ nonorthogonal to $P_1$, we simply fix the order of basis vectors for another subspace in the remaining set of unchosen subspaces and repeat the process. Therefore, we find a set of basis vectors uniquely determined, and we can construct a matrix $M_1$, whose columns are these basis vectors in this canonical order. In the same way, we can also reduce the dimensions for subspaces in $\{Q_1,\dots,Q_K\}$, find a corresponding set of basis vectors and construct matrix $M_2$. 

Finally, we would like to find a unitary $U_c$ such that $U_c^\dagger M_1=M_2$. This can be easily done, by computing the SVD of $M_2M_1^\dagger = U'D'V^{\prime\dagger}$ and choose $U_c=V'U^{\prime\dagger}$. One way to see this is that $U_c^\dagger$ is the solution to the unitary Procrustes problem $\argmin_U\lVert UM_1-M_2\rVert_2^2$ \cite{Gower04Procrustes}. Such $U_C$ exists if and only if the inner products between column vectors of $M_1$ are identical to the inner products between column vectors of $M_2$. Tracing back the iterative reduction and considering the uniqueness of the canonical choice, if these inner products are preserved, then $U_c^\dagger P_iU_c=Q_i,\ \forall i=1,\dots,K$, and therefore $U_c$ is a solution of Eq.~(\ref{eq:unitary equivalence}). In turn, if a unitary satisfying Eq.~(\ref{eq:unitary equivalence}) exists, it must preserve all these inner products by construction. Therefore, $U_c$ satisfying Eq.~(\ref{eq:unitary equivalence}) exists if and only if Algorithm~\ref{alg:unitary control} can successfully identify it.
\end{proof}

\section{Proof of Theorem~\ref{thm:HL Pauli noise}} \label{sec:proof HL Pauli noise}
\begin{proof}
    Without loss of generality, we consider the case where $U_{\theta_0}=V_{\theta_0}=I$; otherwise, a unitary control $U_{\theta_0}^\dagger V_{\theta_0}^\dagger$ can be appended to $T_\theta$. We further choose a basis of linearly independent $n$-qubit Pauli operators (including the identity) $\{P_1',P_2',\dots,P_s'\}$ in $\mathsf{Span}\left\{\langle P_1, P_2, \dots, P_r\rangle\right\}$ for $s\ge r$. Obviously, we have $\mathsf{Span}\left\{\langle P_1, P_2, \dots, P_r\rangle\right\} = \mathsf{Span}\left\{P_1', P_2', \dots, P_s'\right\}$.
    
    As under our assumption $T_{\theta_0}$ is an $n$-qubit Pauli channel, the projection $P$ on the subspace $\mathcal P$ of eigenvectors of $T_\theta^\dagger T_\theta$ with eigenvalues $1$ can be expressed by
    \begin{equation}
        P=\prod_{i=1}^r\left[\frac{1}{2}\left(I + K_i^\theta \otimes K_i^{\theta*}\right)\right]=\prod_{i=1}^r\left[\frac{1}{2}\left(I + P_i \otimes P_i^{*}\right)\right]=\prod_{i=1}^s\left[\frac{1}{2}\left(I + P_i' \otimes P_i^{\prime*}\right)\right].
    \end{equation}
    It is worth mentioning that all $P_i' \otimes P_i^{\prime*}$ commute with each other, so the order in the product does not matter. Besides, in this case $\mathcal P$ is also the subspace of peripheral eigenvectors of $T_{\theta_0}$ itself, so the second condition Eq.~(\ref{eq:unitary equivalence}) in Theorem~\ref{thm:HL with unitary control} is naturally satisfied with trivial identity control $U_c=U_{\theta_0}^\dagger V_{\theta_0}^\dagger=I$.
    
    We now check the first condition Eq.~(\ref{eq:nonvanishing signal}) in Theorem~\ref{thm:HL with unitary control}, by evaluating (at $\theta=\theta_0$)
    \begin{equation} \label{eq:check condition signal not vanishing}
        \begin{aligned}
            P T_\theta^\dagger\dot T_\theta P &= -\mathrm i P T_\theta^\dagger T_\theta \left[H^{(U)} \otimes I - I \otimes H^{(U)*}\right]P  -\mathrm i P T_\theta^\dagger \left[H^{(V)} \otimes I - I \otimes H^{(V)*}\right]T_\theta P\\
            &= -\mathrm i \frac{1}{2^{2s}}\left(I + P_s' \otimes P_s^{\prime*}\right)\cdots\left(I + P_1' \otimes P_1^{\prime*}\right) T_\theta^\dagger T_\theta \left[H^{(U)} \otimes I - I \otimes H^{(U)*}\right] \left(I + P_1' \otimes P_1^{\prime*}\right) \cdots \left(I + P_s' \otimes P_s^{\prime*}\right)\\
            &-\mathrm i \frac{1}{2^{2s}}\left(I + P_s' \otimes P_s^{\prime*}\right)\cdots\left(I + P_1' \otimes P_1^{\prime*}\right) T_\theta^\dagger \left[H^{(V)} \otimes I - I \otimes H^{(V)*}\right] T_\theta  \left(I + P_1' \otimes P_1^{\prime*}\right) \cdots \left(I + P_s' \otimes P_s^{\prime*}\right)\\
            &= -\mathrm i \frac{1}{2^{2s}}\left(I + P_s' \otimes P_s^{\prime*}\right)\cdots\left(I + P_1' \otimes P_1^{\prime*}\right) (H_{\mathrm{tot}} \otimes I - I \otimes H_{\mathrm{tot}}^*) \left(I + P_1' \otimes P_1^{\prime*}\right)\cdots \left(I + P_s' \otimes P_s^{\prime*}\right)\\
            &=-\mathrm i \frac{1}{2^{2s}} \sum_k \alpha_k \left[\left(I + P_s' \otimes P_s^{\prime*}\right)\cdots\left(I + P_1' \otimes P_1^{\prime*}\right)   (Q_k \otimes I - I \otimes Q_k^*)\left(I + P_1' \otimes P_1^{\prime*}\right)\cdots \left(I + P_s' \otimes P_s^{\prime*}\right)\right],
        \end{aligned}
    \end{equation}
having used $H^{(U)}=\mathrm i U_\theta^\dagger \dot U_\theta$, $H^{(V)}=\mathrm i V_\theta^\dagger \dot V_\theta$ and $H_{\mathrm{tot}}=H^{(U)} + H^{(V)} = \sum_k \alpha_k Q_k$. For any $n$-qubit Pauli operator $P_i'$ and any $n$-qubit Pauli operator $Q_k$, if $\{Q_k,P_i'\}=0$, then 
\begin{equation} \label{eq:check condition each Q P}
    \left(I + P_i'\otimes P_i^{\prime*}\right)  (Q_k \otimes I - I \otimes Q_k^*)\left(I + P_i' \otimes P_i^{\prime*}\right)=0;
\end{equation}
if $Q_k=P_i'$, then Eq.~(\ref{eq:check condition each Q P}) also holds. However, if $[Q_k,P_i']=0$, then 
\begin{equation} \label{eq:check condition each Q P commute}
    \left(I + P_i'\otimes P_i^{\prime*}\right)  (Q_k \otimes I - I \otimes Q_k^*)\left(I + P_i' \otimes P_i^{\prime*}\right)=2Q_k \otimes I + 2Q_kP_i'\otimes P_i^{\prime*}-2I\otimes Q_k^*-2P_i'\otimes (Q_kP_i')^*.
\end{equation}
Note that, for any $i=1,\dots,s$, if $Q_k \notin \mathsf{Span}\left\{P_1', P_2', \dots, P_s'\right\}$, then $Q_kP_i' \notin \mathsf{Span}\left\{P_1', P_2', \dots, P_s'\right\}$. Combining it with Eq.~(\ref{eq:check condition each Q P commute}) and requiring $[Q_k, P_i']=0$ for $i=1,\dots,s$ (or equivalently, $[Q_k, P_i]=0$ for $i=1,\dots,r$), each term in the summation over $k$ in the last line of Eq.~(\ref{eq:check condition signal not vanishing}) does not vanish, and in particular, contains Pauli terms $Q_k\otimes I-I\otimes Q_k^*$ (up to a multiplicative factor). Furthermore, in the summation over different $k$ in Eq.~(\ref{eq:check condition signal not vanishing}), such Pauli terms cannot cancel each other as $\{Q_k\}$ are linearly independent. Therefore, we prove that $P T_\theta^\dagger\dot T_\theta P \neq 0$. Finally, by applying Theorem~\ref{thm:HL with unitary control} we complete the proof.
\end{proof}

\emph{Remark}---We remark that Eq.~(\ref{eq:Q not in span}) in Theorem~\ref{thm:HL Pauli noise} can be reformulated as $H_{\mathrm{tot}} \notin \mathsf{Span}\left\{\langle P_1, P_2, \dots, P_r\rangle\right\}$, resembling but stricter than the HNKS condition: $H:=\mathrm i\sum_{k=1}^r K_k^{\theta\dagger} \dot{K}_k^\theta \notin \mathsf{Span}\{K_i^{\theta\dagger} K_j^\theta,\ \forall i,j\}$ \cite{Zhou2021PRXQ}. To see this, we assume that $U_{\theta_0}=V_{\theta_0}=I$ for simplicity; otherwise, appropriate unitary control can be applied to counteract the effects of $U_\theta$ and $V_\theta$. At $\theta=\theta_0$, the ``Hamiltonian'' $H=\mathrm i\sum_{k=1}^r K_k^{\theta\dagger} \dot{K}_k^\theta=H^{(U)}+\sum_{i=1}^r p_iP_iH^{(V)}P_i$, and the ``Kraus span'' becomes a linear span of Pauli operators $\mathsf{Span}\{P_iP_j,\ \forall i,j\}$. Noting that, by Eq.~(\ref{eq:Q commutes}) of Theorem~\ref{thm:HL Pauli noise}, $H_{\mathrm{tot}}$ contains a Pauli term $Q_k$ that commutes with all $P_i$ for $i=1,\dots,r$, Eq.~(\ref{eq:Q not in span}) is equivalent to $H \notin \mathsf{Span}\left\{\langle P_1, P_2, \dots, P_r\rangle\right\}$. This indicates that the Hamiltonian lies outside not only the ``Kraus span'', but also the linear span of the group elements generated by the Pauli noise operators. 

\section{Supplemental examples} \label{sec:supplemental examples}
\subsection{The singularity of HNKS condition} \label{seq:singularity of HNKS}
We first consider a simple but illuminating example, which reveals some subtle issues not captured by the HNKS condition, but made explicit by our theoretical framework. In fact this issue arises from the ill-defined ``Hamiltonian'' in certain cases.

Suppose we want to estimate $\theta$ from $N$ queries to a $\sigma_z$-rotation under the single qubit dephasing noise, described by Kraus operators $K_1=\sqrt{1-p}e^{-\mathrm i\phi \sigma_z/2}$ and $K_2=\sqrt{p}\sigma_z e^{-\mathrm i\phi \sigma_z/2}$. Both $p$ and $\phi$ are functions of the parameter of interest $\theta$. This specific problem is of fundamental interest in the asymptotic theory of quantum channel estimation, because any quantum channel can simulate a logical dephasing channel by using error correction (with ancillae in general) \cite{Zhou2021PRXQ}.

Ref.~\cite[Eq.~(B5)]{Zhou2021PRXQ} showed that this channel has a ``Hamiltonian'' independent of $p$:
$H=\mathrm i \sum_jK_j^\dagger \dot K_j=\frac{\dot{\phi}}{2}\sigma_z$. The authors therefore concluded that HL can only be achieved when (i) $\dot \phi \neq 0$ and (ii) $p=0$ or $1$, according to the HNKS condition. However, our analysis shows that HL can also be achieved if (i) $\dot p\neq 0$ and (ii) $p=0$ or $1$. Specifically, $T_\theta=\sum_i K_i\otimes K_i^*$ has 4 eigenvalues: $\{1,1,(1-2p)e^{-\mathrm i \phi},(1-2p)e^{\mathrm i \phi}\}$, and the eigenvectors are mutually orthogonal. For $\lambda=(1-2p)e^{\mathrm i \phi}$ we have $\left.\dot \lambda\right\rvert_{p=0} = (\mathrm i\dot \phi -2\dot p)e^{\mathrm i \phi}$ and $\left.\dot \lambda\right\rvert_{p=1} = (-\mathrm i\dot \phi -2\dot p)e^{\mathrm i \phi}$, so by Theorem~\ref{thm:condition for HL} we can choose an input state $\rho_0=I/2+\alpha (\ketbra{0}{1}+\ketbra{1}{0})$ (for some $0<\alpha<1/2$) which achieves the HL when (i) $p=0$ or $1$ and (ii) either $\dot\phi\neq 0$ or $\dot p\neq 0$.

\subsection{HL without decoherence-free subspace}\label{sec:example HL beyond QEC}
Existing QEC protocols achieving the HL construct a DFS, where the noise is eliminated but the signal remains to take effect. By QEC the simulated channel is an identity channel on the logical subspace at $\theta=\theta_0$ \cite{Zhou2021PRXQ}. Our analysis, however, opens up new possibilities to achieve the HL without constructing the DFS.

Consider $N$ queries to a quantum channel
described by Kraus operators $K_1=\ketbra{2}{0}$, $K_2=\ketbra{2}{1}$, $K_3=\sqrt{2\theta}\ketbra{2}{2}$, $K_4=\sqrt{1/2-\theta}\ketbra{0}{2}$, and $K_5=\sqrt{1/2-\theta}\ketbra{1}{2}$, where $0\le \theta \le 1/2$ is the parameter of interest. Now $T_\theta=\sum_i K_i\otimes K_i^*$ has 2 nonzero eigenvalues: $\{1,-1+2\theta\}$, with the fixed point $\mathrm{diag}\{1/2-\theta,1/2-\theta,1\}$ and the other eigenmatrix $\mathrm{diag}\{1/2,1/2,-1\}$ (where $\mathrm{diag}\{a_i\}_i$ denotes a diagonal matrix with diagonal elements $\{a_i\}_i$). At $\theta=0$, we can choose an input state $\rho_0=\mathrm{diag}\{1/4,1/4,1/2\}+\alpha \mathrm{diag}\{1/4,1/4,-1/2\}$ for some $-1<\alpha<1$ and $\alpha\neq 0$. Here these two eigenvectors are not orthogonal. By Lemma \ref{lem:asymptotic associated QFI}, the associated QFI of the asymptotic output state at $\theta=0$ exhibits an interesting oscillating behaviour given by
\begin{equation}
    \lim_{N\rightarrow\infty}\frac{\tilde F^Q \left(\tilde \rho_\theta\right)}{N^2} =  \begin{cases}
        \frac{128\alpha^2}{(3\alpha^2+2\alpha+3)^2},&\text{if}\ N \text{ is odd},\\
        \frac{128\alpha^2}{(3\alpha^2-2\alpha+3)^2},&\text{if}\ N \text{ is even}.
    \end{cases}
\end{equation}
The QFI of the ouput state is lower bounded by
\begin{equation}\label{eq:oscillating QFI bound}
    \lim_{N\rightarrow\infty} \frac{F^Q \left(\rho_\theta\right)}{N^2} \ge \lim_{N\rightarrow\infty}\frac{\Tr\left(\rho_\theta^2\right)}{4\lambda_{\mathrm{max}}(\rho_\theta)N^2}\tilde F^Q \left(\tilde \rho_\theta\right) = \begin{cases}
        \frac{8\alpha^2}{(3\alpha^2+2\alpha+3)(1+\alpha)},&\text{if}\ N \text{ is odd},\\
        \frac{16\alpha^2}{(3\alpha^2-2\alpha+3)(1+\alpha)},&\text{if}\ N \text{ is even},
    \end{cases}
\end{equation}
when $1/3<\alpha<1$. However, Eq.~(\ref{eq:oscillating QFI bound}) only gives a lower bound on the actual QFI, so it is unclear whether the QFI also oscillates with $N$ asymptotically. We compute the exact QFI for this example, and find that this oscillating behaviour indeed exists for fairly large $N$, as illustrated in Fig.~\ref{fig:oscillating QFI}. Since $N$ is a finite number in a realistic scenario, we anticipate that this phenomenon may be of practical interest and deserves further investigation.

\begin{figure}[!htbp]
    \centering
    \includegraphics[width=0.6\textwidth]{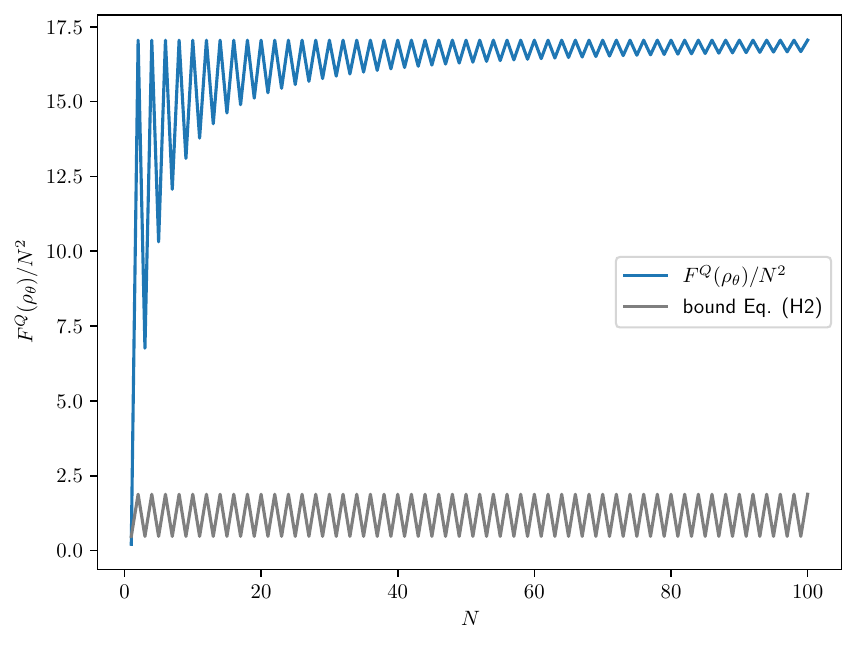}
    \caption{Comparison between the actual normalized QFI $F^Q(\rho_\theta)/N^2$ and the bound Eq.~(\ref{eq:oscillating QFI bound}). The input state is a fixed state $\rho_0=\mathrm{diag}\{1/4,1/4,1/2\}+\alpha \mathrm{diag}\{1/4,1/4,-1/2\}$ for $\alpha=0.9$.}
    \label{fig:oscillating QFI}
\end{figure}

Furthermore, we remark that this quantum channel is irreducible, i.e., having a unique full-rank fixed point, and does not admit a DFS. The HL here is therefore not a result of DFS, but stems from the decay of root-of-unity peripheral eigenvalues corresponding to diagonalizable conserved quantities \cite{Albert2019asymptoticsof}. It is worth noting that the HNKS condition is also ill-defined for this example.

\end{document}